\documentclass{article}
\usepackage{graphicx,amssymb,amsmath,amsthm}
\usepackage{lineno}
\usepackage[usenames,dvipsnames]{xcolor}

\newtheorem{theorem}{Theorem}[section]
\newtheorem{cor}[theorem]{Corollary}
\newtheorem{lemma}[theorem]{Lemma}

\newtheorem{problem}{Problem}

\newtheorem{defn}{Definition}

\newcommand\blfootnote[1]{%
  \begingroup
  \renewcommand\thefootnote{}\footnote{#1}%
  \addtocounter{footnote}{-1}%
  \endgroup
}

\title{Empty Rainbow Triangles in $k$-colored Point Sets}

\author{
Ruy Fabila-Monroy \thanks{Departamento de Matem\'aticas, CINVESTAV. Partially supported by Conacyt of Mexico, Grant 253261.} \thanks{\tt{ruyfabila@math.cinvestav.edu.mx}} \and
Daniel Perz\thanks{Institute for Software Technology, Graz University of Technology, Graz, Austria. Supported by the Austrian Science Fund (FWF): I 3340-N35 \tt{daperz@ist.tugraz.at}} \and
Ana Laura Trujillo-Negrete \footnotemark[1] \thanks{ltrujillo@math.cinvestav.mx}
}

\begin{document}
\maketitle
\blfootnote{\begin{minipage}[l]{0.3\textwidth} \includegraphics[trim=10cm 6cm 10cm 5cm,clip,scale=0.15]{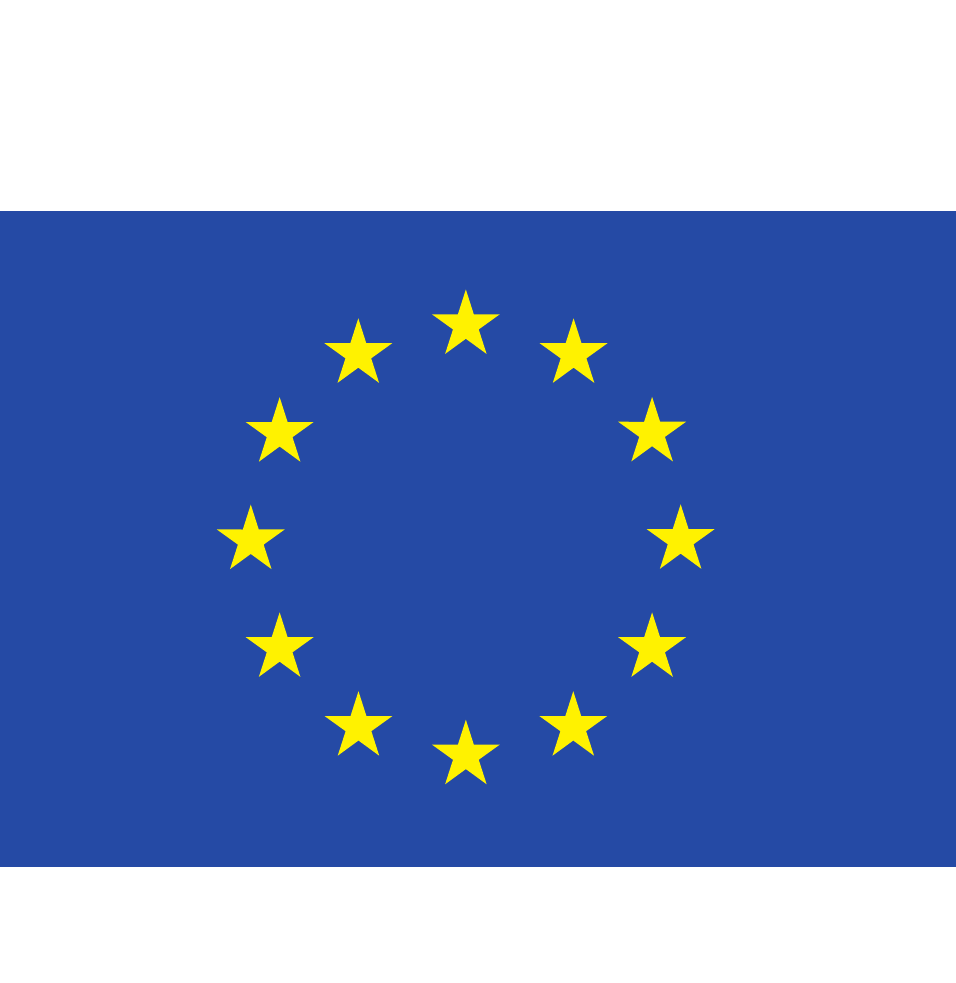} \end{minipage}  \hspace{-2.cm} \begin{minipage}[l][1cm]{0.82\textwidth}
 	  This project has received funding from the European Union's Horizon 2020 research and innovation programme under the Marie Sk\l{}odowska-Curie grant agreement No 734922.
 	\end{minipage}}

\begin{abstract}
Let $S$ be a set of $n$  points in general position in the plane. Suppose that
each point of $S$ has been assigned one of $k \ge 3$ possible colors and that there is
the same number, $m$, of points of each color class. A polygon with vertices on $S$
is empty if it does not contain points of $S$ in its interior; and it is
rainbow if all its vertices have different colors. Let $f(k,m)$ be the minimum
 number of empty rainbow triangles determined by $S$. In this paper we give tight asymptotic bounds 
 for this function. Furthermore, we show that $S$ may not determine an empty rainbow quadrilateral for some arbitrarily large values
of $k$ and $m$.
\end{abstract}

\section{Introduction}
A set of points in the plane is in \emph{general} position if no three of its points are collinear. 
In this paper all sets of points are in general position.
The well known Erd\H{o}s-Szekeres theorem~\cite{happyend} states that for every positive integer $r \ge 3$ there exists a 
positive integer $n(r)$ such that every set of  $n(r)$ (or more points) in the plane
contains the vertices of a convex polygon of $r$ vertices.

Let $S$ be a set of $n$ points in the plane. 
A polygon with vertices on $S$ is said to be \emph{empty} if it does not contain a point of $S$ in its interior.
An \emph{$r$-hole} of $S$ is an empty convex polygon of $r$ sides with vertices on $S$.
In 1978, Erd\H{o}s~\cite{kholes_erdos} asked if for every $r$, every sufficiently large set of points in the plane
contains an $r$-hole. Klein~\cite{happyend} had already noted that every set of $5$ points
contains a $4$-hole. Harboth~\cite{harborth} showed that every set of $10$ points contains
a $5$-hole. Horton~\cite{horton} constructed arbitrarily large sets of points
without $7$-holes. The case for $6$-holes remained open until Nicol\'as~\cite{nicolas}
and Gerken~\cite{gerken}, independently showed that every sufficiently large point set contains a $6$-hole.

Once the existence of $r$-holes for some given $r$ in every sufficiently large point set 
is established, it is natural to ask what is the minimum
number of $r$-holes in every set of $n$ points in the plane. 
Katchalski and Meir~\cite{empty_original} first considered this question for triangles. 
They showed that every set of $n$ points determines
$\Omega(n^2)$ empty triangles and provided an example of a point set determining
$O(n^2)$ empty triangles. The lower and upper bounds on this number have been improved throughout the years~\cite{imre_pavel,dumi,pavel,alfredo,dehnhardt,small_kholes}.
The problem of determining the minimum number of $r$-holes in every set of $n$ points in the plane has also been
considered in these papers. 

Colored variants of these problems where first studied by Devillers, Hurtado, K\'arolyi and Seara~\cite{variants}.
A point set is \emph{$k$-colored} if every one of its points is assigned one of $k$ available colors. 
We say that an $r$-hole on $S$ is \emph{monochromatic}
if all its vertices are of the same color, and that it is \emph{rainbow}\footnote{In~\cite{variants} rainbow $r$-holes
are called \emph{heterochromatic}. We prefer to use ``rainbow'', because this term is used, with this meaning, 
in the more general setting of anti-Ramsey problems. } if all its
vertices are of different colors. Many chromatic variants on problems
regarding $r$-holes in colored points sets have been studied 
since; see ~\cite{almost,mono_pach, mono_us, non_convex_quad, balanced_4_holes,almost_quad,balanced_6_holes,
clemens_4_gon,brass_4_gon,friedman_4_gon,gulik_4_gon}.
In particular, Aichholzer, Fabila-Monroy, Flores-Pe\~naloza, Hackl, Huemer, and Urrutia showed
that every $2$-colored set of $n$ points in the plane determines $\Omega(n^{5/4})$ empty
monochromatic triangles~\cite{mono_us}. This was later improved to $\Omega(n^{4/3})$ by Pach and T\'oth~\cite{mono_pach}.
The current best upper bound on this number is $O(n^2)$ and this is conjectured to be
the right asymptotic value.

In this paper we consider the problem of counting the number of empty
rainbow triangles in $k$-colored point sets in which there are the same 
number, $m$, of points of each color class.  Let $f(k,m)$ be the minimum number of empty
rainbow triangles in such a point set. We give the following tight asymptotic bound
for $f(k,m)$. 
\begin{theorem}\label{thm:main}
\begin{equation*}
  f(k,m) = \left\{
\begin{array}{lr}
 \Theta(k^2m)  & \textrm{ if } m < k,\\
 
  \Theta(k^3) & \textrm{ if }  m \ge k.
\end{array} \right. 
\end{equation*}
\end{theorem}

Note that in contrast to the number of empty monochromatic triangles, 
the number of empty rainbow triangles does not necessarily  grow with the number of points.

\section{Lower Bound}

%
\begin{proof}[Proof of the lower bound in Theorem~\ref{thm:main}]
Let $S$ be a $k$-colored set of points with $m$ points of each color class. Without loss of generality
assume that no two points of $S$ have the same $x$-coordinate. 
Assume that the set of colors is $\{1,\dots,k\}$.
For each $1 \le i \le k$, let $p_i$ be the leftmost point of color $i$.
Without loss of generality assume that when sorted by $x$-coordinate
these points are $p_1,\dots, p_k$.

Let $1 \le i \le k$ and let $r_i:=\min\{i,m\}$.
We show that there are at least $(r_i^2-3r_i+2)/2$ 
empty rainbow triangles having a point of color $i$ as its rightmost point.
Let $q_1:=p_i, q_2,\dots, q_{r_i-2}$ be the first $r_i-2$ points of color
$i$ when sorted by $x$-coordinate. For each $1 \le j \le r_i-2$ do the following. Sort the points
of $S$ to the left of $q_j$ counterclockwise by angle around $q_j$. Note that any 
two consecutive points in this order define an empty triangle
with $q_j$ as its rightmost point. Since the points $p_1,\dots,p_{i-1}$
are to the left of $q_j$, there are at least $i-2$ of these empty triangles
such that the first point is of a color $l$ distinct from $i$, and the 
next point is of a color distinct from $l$. Furthermore, for at least
$(i-2)-(j-1)=i-j-1$ of these triangles the next point is not of color $i$; thus, they are rainbow.  We have at least
\begin{equation}\label{eq:rightmost}
 \sum_{j=1}^{r_i-1} i-j-1=\frac{(r_i-1)(2i-r_i-2)}{2} 
 \end{equation}
empty rainbow triangles with a point of color $i$ as its rightmost point.
If  $i \le m$ then the right hand side of (\ref{eq:rightmost}) is equal to 
\[\frac{i^2-3i+2}{2}\]

Thus, if $m \ge k$ then $S$ determines at least
\[\sum_{i=3}^{k}\frac{i^2-3i+2}{2} =\frac{1}{6}k^3-\frac{1}{2}k^2+\frac{1}{3}k=\Omega(k^3) \]
empty rainbow triangles; and
if $m < k$ then $S$ determines at least
\begin{align}
 & \sum_{i=3}^{k}\frac{(r_i-1)(2i-r_i-2)}{2} \nonumber \\
 & = \sum_{i=3}^{m}\frac{i^2-3i+2}{2}
 + \sum_{i=m+1}^{k}\frac{(m-1)(2i-m-2)}{2}\nonumber \\
 &=\frac{1}{2}k^2m-\frac{1}{2}km^2+\frac{1}{6}m^3-\frac{1}{2}k^2+\frac{1}{2}k-\frac{1}{6}m \nonumber \\
 & = \Omega(k^2m)+\Omega(km^2+m^3) \nonumber \\
 &=  \Omega(k^2m) \nonumber
\end{align}
empty rainbow triangles
\end{proof}

\section{Upper Bound}
In this section we construct a $k$-colored point set which provides our upper bounds for $f(k,m)$.

\subsection{The Empty Triangles of the Horton Set}

As a building block for our construction we use Horton sets~\cite{horton};
in this section we characterize the empty triangles of the Horton set.
Let $H$ be a set of $n$ points in the plane with no two points
having the same $x$-coordinate; sort
its points by their $x$-coordinate so that 
$H=\{p_0, p_1,\dots, p_{n-1}\}$. Let $H_{0}$ be the subset of the
even-indexed points of $H$, and  $H_{1}$ be the subset of the 
odd-indexed points of $H$. That is, $H_{0}=\{p_0, p_2,\dots\}$
and $H_{1}=\{p_1, p_3,\dots\}$. Let $X$ and $Y$ be two finite sets of points in the plane.
We say that $X$ is \emph{high above} $Y$ if: every line determined by two points in $X$ is above every point in $Y$; and
every line determined by two points in $Y$ is below every point in $X$.

\begin{defn} \label{def:mat} 
$H$ is a \textbf{Horton set} if 
\begin{enumerate}
  \item $|H|=1$; or

  \item $|H|\ge 2$; $H_{0}$ and $H_{1}$ are Horton sets; and
  $H_{1}$ is high above $H_{0}$.
\end{enumerate}
\end{defn}

Assume that $H$ is a Horton set. We say that an edge $e:=(p_i,p_j)$ is a
\emph{visible edge} of $H$ if one of the following two conditions are met.
\begin{itemize}
 \item Both $i$ and $j$ are even and for every even $i < l < j$, the point 
 $p_l$ is below the line passing through $e$. In this case we say that
 $e$ is \emph{visible from above}.
 
 \item Both $i$ and $j$ are odd and for every odd $i < l < j$, the point 
 $p_l$ is above the line passing through $e$. In this case we say that
 $e$ is \emph{visible from below}.
\end{itemize}

\begin{lemma}\label{lem:visible_edges}
 The number of visible edges of $H$ is less than $2n$.
\end{lemma}
\begin{proof}
Let $s:=100\cdots0$ be a binary string starting with a $1$ and followed by a trail of $0$'s of length at most $\lceil \log_2(n)\rceil$.
Every consecutive pair of points of $H_{s}$ defines a visible edge
from below of $H$. Moreover, all visible edges from below of $H$ are of this form, for some $s$. 
Note that  $|H_{s}|\le n/2^{|s|}+1$. A similar analysis holds for the edges visible from above of $H$,  using the binary strings starting with a $0$ and followed by a trail
of $1$'s of length at most $\lceil \log_2(n)\rceil$.
The number of visible edges  of $H$ is at most
\[2 \sum_{i=1}^{\lceil \log_2(n)\rceil } \frac{n}{2^{i}}<2n.\]
\end{proof}

The visible edges of $H$ allows to characterize its empty triangles recursively as follows.
\begin{lemma}\label{lem:splitted_triangs}
 Let $p_i,p_j$ and $p_l$ be the vertices of a triangle $\tau$ of $H$ such that
 either
 \begin{itemize}
  \item  $(p_i,p_j)$ is an edge visible from below 
 and $p_l \in H_0$; or 
 
  \item  $(p_i,p_j)$ is an edge visible from above 
 and $p_l \in H_1$.
 \end{itemize}
 Then $\tau$ is empty. Moreover, every empty triangle of $H$ with at least
 one vertex in each of $H_0$ and $H_1$ is of one these forms.
\end{lemma}
\begin{proof}
 If $\tau$ is such a triangle then its emptiness follows from the definition
 of the Horton set. Suppose now that $\tau:=p_ip_jp_l$ is an empty triangle of $H$ with
 $p_i, p_j \in H_0$ and $p_l \in H_1$, or $p_i, p_j \in H_1$ and $p_l \in H_0$.
 Then, for $\tau$ to be empty, $p_ip_j$ must be an edge visible from above (resp. below).
\end{proof}

We can now get a good upper bound on the number of empty triangles of $H$.
\begin{cor}\label{cor:empty_triangles}
The number of empty triangles of $H$ is at most $2n^2$. 
\end{cor}
\begin{proof}
 Let $T(n)$ be the number of empty triangles in a Horton set 
 of $n$ points. Then $T(n)$ is equal to the number of empty triangles
 with at least one vertex in each of $H_0$ and $H_1$, plus the number of
 empty triangles with all their vertices in $H_0$ or all their vertices in $H_1$. By the definition
 of Horton sets and Lemma~\ref{lem:splitted_triangs}
 we have that 
 \[ T(n) < T \left( \left \lceil \frac{n}{2} \right \rceil \right )+T \left( \left \lfloor \frac{n}{2} \right \rfloor \right )+n^{2} \le 2n^2.\]
\end{proof}

\subsection{Blockers}

Our strategy is to start with a Horton set $H$ of $k$ points and replace each point $p_i$ of $H$
with a cluster $C_i$ of $m$ points. All of the points of $C_i$ are of the same color
and are at a distance of at most some $\varepsilon$ from $p_i$. We choose
$\varepsilon$ to be arbitrarily small. Let  $S$ be the resulting
set. Note that every rainbow triangle of $S$ must have all its vertices in different
clusters. Moreover, since each $C_i$ is arbitrarily close to $p_i$ we have the following.
If $\tau$ is an empty triangle of $S$ with vertices in different clusters $C_i$, $C_j$ and $C_l$
then $p_i$,$p_j$ and $p_l$ are the vertices of an empty triangle in $H$. In principle,
this gives $m^3$ empty rainbow triangles in $S$ per empty triangle of $H$. However, we can place
the points within each cluster in such a way so that only very few of these triangles
are actually empty.

Let $p_i \in H$, and $r:=\min\{\left \lceil \log_2(k) \right\rceil+2,\lceil m/2  \rceil\}$. In what follows we iteratively define
real numbers \[\varepsilon = \varepsilon_1 > \varepsilon_2 > \dots >\varepsilon_{r+1} >0;\]
in the process we also place a subset $B_i$ of points of $C_i$ at some of these distances; 
we refer to the points in $B_i$ as \emph{blockers}.
For $t=1\dots,r$
suppose that $\varepsilon_t$ has been defined and possibly some
points of $B_i$ have been placed. Consider every pair of points $p_j,p_l \in H$ distinct from $p_i$.
Let $q \in B_i$ be at distance $\varepsilon_t$ or more from $p_i$ and such that $q$ is in the interior
of every triangle with vertices $p_i,p_j'$ and $p_l'$, where $p_j'$ and $p_l'$ 
are at a distance of at most $\varepsilon$ of $p_j$ and $p_l$, respectively.
Let $\tau$ be the triangle with vertices $p_i',p_j'$ and $p_l'$,
where $p_i'$ is any point at a distance of at most $\varepsilon_{t+1}$ from $p_i$.
We define $\varepsilon_{t+1} < \varepsilon_{t}$ small enough so that every such
$q$ is in the interior of every such $\tau$.
We say that $q$ \emph{blocks} the triangle with vertices  $p_i',p_j'$ and $p_l'$.

We construct $B_i$ iteratively as follows.
We say that a blocker point at distance $\varepsilon_t$ from $p_i$ is at \emph{layer $t$}. 
Let $s_0,\dots,s_{r'}$ be the binary strings such that:
\begin{itemize}
 \item[$1)$] $s_0=\emptyset$;
 
 \item[$2)$]  for every $0 \le t < r'$, $s_{t+1}=s_t0$ or  $s_{t+1}=s_t1$; and
 
 \item[$3)$] $H_{s_{r'}}=\{p_i\}$.
\end{itemize}
By $2)$ and $3)$ we have that $p_i \in H_{s_t}$ for every $0 \le t  \le r'$. Note that 
$r' \le \left \lceil \log_2(k) \right\rceil$.

Sort the points of $H \setminus\{p_i\}$ counterclockwise
by angle around $p_i$.
For every $t=0,\dots,r'-2$ and as long as we have placed at most $m-2$ blocker points, we place two blocker
points at a distance from $\varepsilon_{t+1}$ from $p_i$ as follows.
\begin{itemize}
 \item  Suppose that $s_{t+1}=s_t0$. Place one blocker point just after the leftmost point of $H_{s_t1}$
 in order by angle around $p_i$; place another blocker point just before the rightmost point of $H_{s_t1}$
 in order by angle around $p_i$, as depicted in Figure~\ref{fig:blocker}(a).
 
 \item Suppose that $s_{t+1}=s_t1$. Place one blocker point just after the leftmost point of $H_{s_t0}$
 in order by angle around $p_i$; place another blocker point just before the rightmost point of $H_{s_t0}$
 in order by angle around $p_i$.
\end{itemize}
Let $B_i' \subset B_i$ the set of these blocker points.
If $|B_i'|<m$ then we proceed to place the remaining points of $C_i$. 
 If $m-|B_i'| < k$ then place the remaining
points of $C_i$ in any way at a distance of at most $\varepsilon_{r}$ of $p_i$; in this case we have that $B_i=B_i'$.
Suppose that $m -|B_i'| \ge k$. For every $t=0,\dots,r'-1$ we place additional blocker points as follows.
\begin{itemize}
 \item  If $s_{t+1}=s_t0$ then place a blocker point, at a distance of $\varepsilon_r$ from $p_i$, between any two consecutive
 points of $H_{s_t1}$ in order by angle around $p_i$; see Figure~\ref{fig:blocker}(b).
 
 \item  If $s_{t+1}=s_t1$ then place a blocker point, at a distance of $\varepsilon_r$ from $p_i$, between any two consecutive
 points of $H_{s_t0}$ in order by angle around $p_i$.
\end{itemize}
Let $B_i'' \subset B_i$ the set of these blocker points. No more blocker points are added and $B_i=B_i'\cup B_i''$.
If $m > |B_i|= |B_i'|+|B_i''|$ then place the remaining
points of $C_i$ in any way at a distance of at most $\varepsilon_{r+1}$ from $p_i$.

\begin{figure}
\centering
\includegraphics{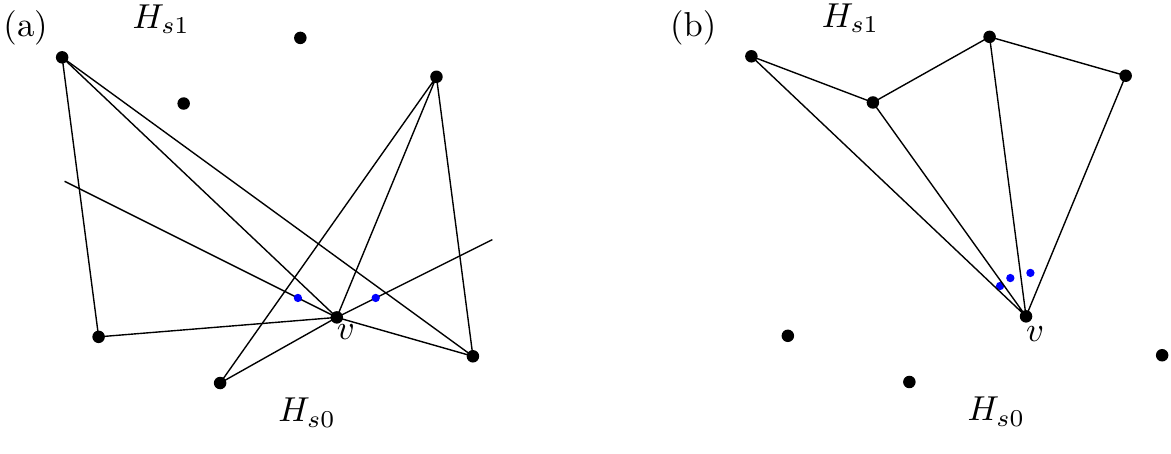}
\caption{(a) Blocking triangles with one point in $H_{s0}$ and one point in $H_{s1}$ (b) Blocking triangles with two vertices in $H_{s1}$.}
\label{fig:blocker}
\end{figure}

We are now ready to prove our upper bounds on $f(k,m)$.
%
\begin{proof}[Proof of the upper bound in Theorem~\ref{thm:main}]
We count the number of empty rainbow triangles determined by $S$ as constructed above.
To every empty triangle $\tau'$ of $S$ with vertices $p_i'\in C_i$, $p_j'\in C_j$ 
and $p_l' \in C_l$, we assign the empty triangle $\tau$ of $H$ with vertices $p_i$,
$p_j$ and $p_l$.
 Let $s$ be the binary string such that the vertices of $\tau$ are contained in $H_s$ but not in
 $H_{s0}$ and $H_{s1}$. We say that $\tau$ is in layer $|s|+1$.
 Without loss of generality suppose that $p_j$ and $p_l$
 are both contained in $H_{s0}$ or are both contained in $H_{s1}$.
 
 If $\tau'$ contains a blocker point of each of $B_j'$ and $B_l'$ then these blocker points
 are at layer $|s|+1$. In this case there at most $2(|s|+1)$ possible choices for each of $p_i'$ and $p_j'$.
 Otherwise, $m < 2(|s|+1)$ and there are at most $m$ possible choices for each of  $p_j'$ and $p_l'$.
 If $m \ge k + 2\left \lceil \log_2(k) \right\rceil$ then $\tau'$ contains a point
 from $B_i''$; and there at most  $k + 2\left \lceil \log_2(k) \right\rceil$ possible
 choices for $p_i'$. Otherwise, $m <  k + 2\left \lceil \log_2(k) \right\rceil$
 and there at most $m$ possible choices for $p_i'$.
 Summarizing, $\tau$ is assigned to at most the following number of empty rainbow triangles of $S$:
 \begin{equation*}
\begin{array}{ll}
  m^{3}  & \textrm{ if } m \le  2|s|+1; \\
 
  4(|s|+1)^2 m & \textrm{ if } 2|s|+1 < m <  k + 2\left \lceil \log_2(k) \right\rceil; \textrm{ and}\\
  
  4(|s|+1)^2 \left (k +  2\left \lceil \log_2(k) \right\rceil \right ) & \textrm{ if } m \ge k + 2\left \lceil \log_2(k) \right\rceil.
\end{array} 
\end{equation*}

 By Lemma~\ref{lem:splitted_triangs}, $(p_j,p_l)$ is a visible edge of $H_s$. 
 Since $|H_s| \le \left \lceil k/2^{|s|}  \right \rceil$, by Lemma~\ref{lem:visible_edges} 
 there are at most 
 $2\left \lceil k/2^{|s|}  \right \rceil \left \lceil k/2^{|s|+1}  \right \rceil \le 8 \left (k^2/2^{2|s|} \right )$ empty
 triangles in $H_s$. Thus, for every $1 \le t \le \lceil \log_2(k) \rceil$ 
 there at most $2^{t-1}8 \left (k^2/2^{2 (t-1)} \right ) =8k^2/2^{t-1}$
 empty triangles in $H$ at layer $t$. Let $m':=\min\{m, k +  2\left \lceil \log_2(k) \right\rceil \}$. Therefore, the number of empty rainbow
 triangles determined by $S$ is at most
 \begin{equation}
  \sum_{t=1}^{\lfloor m'/2 \rfloor} 4t^2 m'  \left (\frac{8k^2}{2^{t-1}} \right ) + \sum_{t=\lfloor m/2 \rfloor+1}^{\lceil \log_2(k) \rceil}{m}^3   \left ( \frac{8k^2}{2^{t-1}} \right ), \label{eq:triangs}
 \end{equation}
 where the second term is set to $0$ if $\lfloor m/2 \rfloor > \lceil \log_2(k) \rceil$. 
 
 If $m'=m$ then (\ref{eq:triangs}) is at most
 \begin{align}
    & \sum_{t=1}^{\lfloor m/2 \rfloor} 4t^2 m  \left (\frac{8k^2}{2^{t-1}} \right ) + \sum_{t=\lfloor m/2 \rfloor+1}^{\lceil \log_2(k) \rceil} 4t^2m   \left ( \frac{8k^2}{2^{t-1}} \right ),   \nonumber \\
    &= \sum_{t=1}^{\lceil \log_2(k) \rceil} 4t^2 m  \left (\frac{8k^2}{2^{t-1}} \right ) \nonumber \\
    &= 32 k^2 m  \sum_{t=1}^{\lceil \log_2(k) \rceil} \left (\frac{t^2}{2^{t-1}} \right ) \nonumber \\
		&\leq 384 k^2m \nonumber \\
    & =O(k^2 m).\nonumber
 \end{align}
If $m' = k+ 2\left \lceil \log_2(k) \right\rceil $ then (\ref{eq:triangs}) is at most
 \begin{align}
    & \sum_{t=1}^{\lceil \log_2(k) \rceil} 4t^2 \left( k+ 2\left \lceil \log_2(k) \right\rceil \right )  \left (\frac{8k^2}{2^{t-1}} \right )  \nonumber \\
    &= 32 k^2 \left( k+ 2\left \lceil \log_2(k) \right\rceil \right )  \sum_{t=1}^{\lceil \log_2(k) \rceil} \left (\frac{t^2}{2^{t-1}} \right ) \nonumber \\
		&\leq 384 k^2 \left( k+ 2\left \lceil \log_2(k) \right\rceil \right ) \nonumber \\
    & =O(k^3).\nonumber
 \end{align}
 Therefore, \begin{equation*}
  f(k,m) = \left\{
\begin{array}{ll}
  O(k^2m ) & \textrm{ if }   m \le k,\\
  
  O(k^3) & \textrm{ if } m > k
\end{array} \right. 
\end{equation*}
\end{proof}

\section{Empty Rainbow Quadrilaterals}
A natural generalization is to consider empty rainbow polygons;
we construct a $k$-colored point set with the same number of points 
in each color class and that does not determine an empty rainbow quadrilateral.
First, we observe the following.
\begin{lemma}\label{sec4:bas_idea}
The point set depicted in Figure~\ref{4gon:detail} does not determine an empty rainbow quadrilateral.
\end{lemma}
\begin{proof}
Let $\tau$ be a rainbow quadrilateral of the point set depicted in Figure~\ref{4gon:detail}.
Note that $\tau$ must have $A,B$ and $C$ as vertices.
Thus, at least two of $AB$, $AC$ and $BC$ are sides of $\tau$.
Assume without loss of generality that $AB$ and $AC$ are sides of $\tau$.
If the fourth vertex of $\tau$ is not one of the red points near $A$ then these points are inside $\tau$;
and $\tau$ is not empty. If the fourth vertex of $\tau$ is one of the red points near  $A$
then by construction the other red point near $A$ is inside $\tau$; and again $\tau$ is not empty. 
\end{proof}

\begin{figure}[h]
\center
\includegraphics{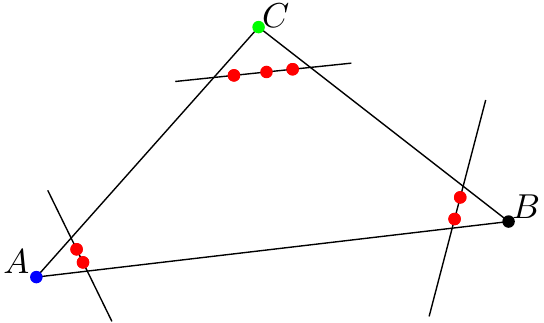}
\caption{A colored point set without an empty rainbow quadrilateral}
\label{4gon:detail}
\end{figure}

%


We use Lemma~\ref{sec4:bas_idea} to construct our point set.
First we take a regular $(k-1)$-gon, $P$, with vertices $p_1,\dots,p_{k-1}$; we replace every point $p_i$ with a cluster $C_i$ of $m$ points of color $i$.
Let $P'$ be a copy of $P$ with vertices $p'_1,\dots,p'_{k-1}$, which is rotated by $\frac{2\pi}{2(k-1)}=\frac{\pi}{k-1}$. So $p_1,p'_1,p_2,\dots,p_{k-1},p'_{k-1}$ form a regular $2(k-1)$-gon. 
Let $\varepsilon$ be sufficiently small.  
For every $1 \le i \le k-1$, we place the points of $C_i$ at a distance of at most $\varepsilon$ from $p_i$. 
For $1\leq i \leq k-1$ ($p'_0=p'_{k-1}$) we place at least $2(k-3)$ points of color $k$ arbitrarily close to the line segment $p'_{i-1}p'_{i}$ 
and so that the following holds. 
Let $q_1$ and $q_2$ be any two consecutive points of $P$ distinct from $p_i$. 
In the triangle with vertices $p_i,q_1$ and $q_2$ there are at least two points on $p'_{i-1}p'_{i}$ of color $k$.
Furthermore, these points are at a distance of at least $\varepsilon$ to the lines $\overline{p_i q_1}$ and $\overline{p_i q_2}$.
Note that $m\geq 2k^2-8k+6$. Let $C_k$ be the set of points of color $k$ in this construction.
The construction for $k=6$ is depicted in Figure~\ref{4gon:construction}. 

\begin{figure}[h]
\center
\includegraphics{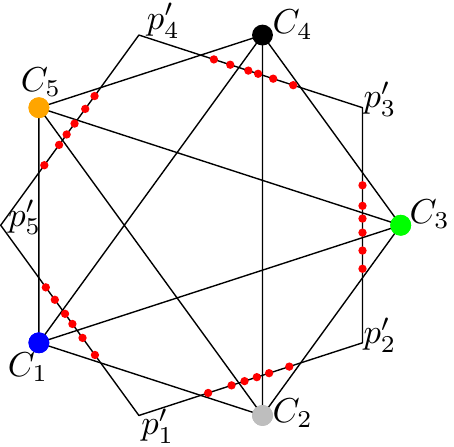}
\caption{A construction of a $6$-colored point set without empty rainbow $4$-gons; the clusters $C_i$ are drawn enlarged.}
\label{4gon:construction}
\end{figure}

\begin{theorem}
The set $\bigcup_{i=1}^{k} C_i$ is a $k$-colored point set, with the same number of
points of each color class, that  does not determine an empty rainbow quadrilateral.
\end{theorem}

\begin{proof}
Let  $\tau$ be a rainbow triangle with vertices $q_i \in C_i$, $q_j\in C_j$ and $q_l \in C_l$. 
We show that $\tau$ has a structure like the point set depicted in Figure~\ref{4gon:detail}.
Consider the points with color $k$ near the line segment $p'_{i-1}p'_{i}$ and that are between the line segments $q_i q_j$ and $q_i q_l$.
By construction, there are at least two points of color $k$ near $p'_{i-1}p'_{i}$ between $p_i p_j$ and $p_i p_l$.
Since these points are at a distance of at most $\varepsilon$ to either of these lines, they are also between $q_i q_j$ and $q_i q_l$.
By the same argument there are points of color $k$ on $p'_{j-1}p'_j$ and $p'_{l-1}p'_l$ inside $\tau$.
Therefore, by Lemma~\ref{sec4:bas_idea}, the point set $\bigcup_{i=1}^{k} C_i$ does not determine an empty rainbow quadrilateral.
\end{proof}

%
%
%

 The construction  described above determines many monochromatic quadrilaterals. This leads us to the following question.
\begin{problem}
Does every sufficiently large $k$-colored ($k \ge 4$) point set with the same 
 number of points in each color class determines an
 empty rainbow quadrilateral or an empty monochromatic quadrilateral?
\end{problem}


\small
\bibliographystyle{abbrv} \bibliography{rainbow}


\end{document}